\documentclass[12pt]{amsart}

\usepackage{amsmath,amsthm,amssymb,amsfonts,verbatim}
\usepackage{graphicx}
\usepackage[colorlinks, linkcolor=blue,  anchorcolor=blue,
            citecolor=blue
            ]{hyperref}
\usepackage[hmargin=1.2in,vmargin=1.2in]{geometry}

\title[Optimal LRC via elliptic curves]{Optimal locally repairable codes via elliptic curves}
\author{Xudong Li}\address{Lab of Security Insurance Cyberspace  and School of Science, Xihua University, Chengdu, China
225002}\email{lixudong73@163.com}
\author{Liming Ma}\address{School of Mathematical Sciences, Yangzhou University, Yangzhou, China
225002}\email{lmma@yzu.edu.cn}
\author{Chaoping Xing} \address{Division of Mathematical Sciences, School of Physical Mathematical Sciences,
Nanyang Technological University, Singapore
637371}\email{xingcp@ntu.edu.sg}
\date{}

\newtheorem{lemma}{Lemma}[section]
\newtheorem{theorem}[lemma]{Theorem}

\newtheorem{prop}[lemma]{Proposition}
\newtheorem{ex}[lemma]{Example}

\newtheorem{defn}{Definition}

\theoremstyle{remark}
\newtheorem{rmk}{Remark}

\renewcommand{\epsilon}{\varepsilon}
\renewcommand{\le}{\leqslant}
\renewcommand{\ge}{\geqslant}

\def\Gal{{\rm Gal}}

\def\PP{\mathbb{P}}

\def\F{\mathbb{F}}

\def \fE {\mathfrak{E}}

\def \mL {\mathcal{L}}

\def \mP {\mathcal{P}}

\def \Xi {{X^{[i]}}}

\newcommand{\Ga}{\alpha}
\newcommand{\Gb}{\beta}
\newcommand{\Gg}{\gamma}

\newcommand{\Gs}{\sigma}

\def \bc {{\bf c}}

\def \bo {{\bf 0}}

\def\supp {{\rm supp }}

\def\mG{{\mathcal G}}

\def\Aut {{\rm Aut }}

\def\LRC {{\rm locally repairable code\ }}

\def\Gal{{\rm Gal}}

\begin{document}

\maketitle

\begin{abstract} Constructing locally repairable codes achieving Singleton-type bound (we call them optimal codes in this paper) is a challenging task and has attracted great attention in the last few years. Tamo and Barg \cite{TB14} first gave a breakthrough result in this topic by cleverly considering subcodes of Reed-Solomon codes. Thus, $q$-ary optimal locally repairable codes from subcodes of Reed-Solomon codes given in \cite{TB14} have length upper bounded by $q$. Recently, it was shown through extension of construction in \cite{TB14}   that  length of $q$-ary optimal locally repairable codes can be $q+1$ in \cite{JMX17}. Surprisingly  it was shown in \cite{BHHMV16} that, unlike classical MDS codes, $q$-ary optimal locally repairable codes could have length bigger than $q+1$. Thus, it becomes an interesting and challenging problem to construct $q$-ary optimal locally repairable codes of length bigger than $q+1$.

In the present paper, we make use of rich algebraic structures of elliptic curves to construct a family of $q$-ary optimal locally repairable codes of length up to $q+2\sqrt{q}$. 
It turns out that
locality of our codes can be as big as $23$ and distance can be linear in length.
\end{abstract}

\section{Introduction}\label{sec:1}
Because of recent applications to distributed storage systems, people have introduced a new class of block codes, i.e, locally repairable codes and they have attracted great attention of researchers \cite{HL07,GHSY12,PKLK12,SRKV13,FY14,PD14,TB14,TPD16,BTV17}. A local repairable code is just a  block code with an additional parameter called {\it locality}.  For a locally repairable code $C$ of length $n$ with $k$ information symbols and locality $r$ (see the definition of locally repairable codes in Section \ref{subsec:2.1}), it was proved in \cite{GHSY12} that the minimum distance $d(C)$ of $C$ is upper bounded by
 \begin{equation}\label{eq:x1}
 d(C)\le n-k-\left\lceil \frac kr\right\rceil+2.
 \end{equation}
 The bound \eqref{eq:x1} is called the Singleton-type bound for locally repairable codes and was proved by extending the arguments in the proof of the classical Singleton bound on codes.  In this paper, we refer an optimal locally repairable code to a block code achieving the bound \eqref{eq:x1}.

\subsection{Known results}
Construction of optimal  locally repairable codes, i.e., block codes archiving the bound \eqref{eq:x1} is of both theoretical interest and practical importance. This is a challenging task and has attracted great attention in the last few years. In literature, there are a few constructions available and some classes of  optimal  locally repairable codes are known. A class of codes constructed earlier and known as pyramid codes \cite{HCL} are shown to be codes that
are optimal.  In \cite{SRKV13},  Silberstein {\it et al}  proposed a two-level construction based on the  Gabidulin codes combined with a single parity-check $(r+1,r)$ code. Another construction \cite{TPD16} used two layers of MDS codes, a Reed-Solomon code and a special $(r+1,r)$ MDS code. A common shortcoming of these constructions relates to the size of the code alphabet which in all the papers is an exponential function of the code length, complicating the implementation. There  was an earlier construction of optimal locally repairable codes given in \cite{PKLK12} with  alphabet  size comparable to code length. However, the construction in \cite{PKLK12} only produces  a specific value of the length $n$, i.e., $n=\left\lceil \frac kr\right\rceil(r+1)$. Thus, the rate of the code is very close to $1$. There are also some existence results given in \cite{PKLK12} and \cite{TB14} with less restriction on locality $r$. But both results require large alphabet which is an exponential function of the code length.

 A recent breakthrough construction was given  in \cite{TB14}. This construction naturally generalizes Reed-Solomon construction which relies on the alphabet of cardinality comparable to the code length $n$. The idea behind the construction is very nice. The only shortcoming of this construction is restriction on locality $r$. Namely,  $r+1$ must be a divisor of either $q-1$ or $q$, or $r+1$ is equal to a product of a divisor of $q-1$ and a divisor of $q$ for certain $q$, where $q$ is the code alphabet. This construction was extended via automorphism group of rational function fields by Jin, Ma and Xing \cite{JMX17} and it turns out that there are more flexibility on locality and the code length can be $q+1$.

 Based on the classical MDS conjecture, one should wonder if $q$-ary optimal  locally repairable codes can have length bigger than $q+1$. Surprisingly, it was shown in \cite{BHHMV16} that there exist $q$-ary optimal locally repairable codes of length exceeding  $q+1$. Although only few $q$-ary optimal locally repairable codes with length bigger than $q+1$ were produced in \cite{BHHMV16}, it paves a road for people to continue search for such codes.

\subsection{Our main results}
In this paper, we make use of rich algebraic structures of elliptic curves over finite field to construct $q$-ary optimal locally repairable codes with length bigger than $q+1$. More precisely speaking, we take a subgroup $\mG$ of the automorphism group $\Aut(\fE/\F_q)$ of an elliptic curve $\fE$ over a finite field $\F_q$, then consider the subfield $F$ of the elliptic function field $\F_q(\fE)$ whose elements are fixed by $\mG$. By carefully choosing functions from both $F$ and $\F_q(\fE)$ and mixing them together, we can define a subcode of an algebraic geometry code by taking these function as evaluation functions. It can be shown that this subcode is an optimal locally repairable code. As an elliptic curve over $\F_q$ has more than $q+1$ points and the length of the code is almost equal to the number of points, thus we obtain a $q$-ary optimal locally repairable code of length bigger than $q+1$. Our main result of this paper can be summarized below.

\begin{theorem}\label{thm:1.1} Let $q=p^a$ for a prime $p$ and an even number $a>0$. If $p=3$ or $p\equiv  2\pmod{3}$, then there exists an optimal $q$-ary $[n=3\ell, k=2t+1,d=n-3t]$ locally repairable code with locality $2$ for any $0\le t<\ell\le \left\lfloor\frac{q+2\sqrt{q}}3\right\rfloor$.
\end{theorem}

\begin{theorem}\label{thm:1.2}
Let $q=p^a$ for a prime $p$ and an even number $a>0$. Then  there exists an optimal $q$-ary $[n=(r+1)\ell, k=r(t-1)+1,d=n-(t-1)(r+1)]$ locally repairable code with locality $r$ for any integers $t$ and $\ell$ satisfying $1\le t<\ell\le\left\lfloor\frac{q+2\sqrt{q}-r-2}{r+1}\right\rfloor$ if $p$ and $r$ satisfy one of the followings.
\begin{itemize}
\item[{\rm (i)}] $r=3$, $p=2$ or $p\equiv 3\pmod{4}$.
\item[{\rm (ii)}] $r=5$, $p=3$ or $p\equiv 2\pmod{3}$.
\item[{\rm (iii)}] $r=7$, $p=2$.
\item[{\rm (iv)}] $r=11$, $p=2$ or $3$.
\item[{\rm (v)}] $r=23$, $p=2$.
\end{itemize}
\end{theorem}
\begin{rmk} \begin{itemize}
\item [(i)] In the paper \cite{BHHMV16}, a few optimal locally repairable codes such as $4$-ary $[18,11,2]$ code with locality $2$, $5$-ary $[24,17,3]$ code with locality $3$ etc are found based on various surfaces. The $q$-ary optimal locally repairable codes  given in \cite{BHHMV16} have length around $q^2$, but small distance $d$ and locality $r$ such as $d=3$ and $r=2,3, 4$.  The optimal codes in this paper have length slightly bigger than $q+1$. However, the minimum distance of our code can be linear in length and locality can be as large as $23$.
\item[(ii)]
Although we only state the result for codes over $\F_{q^a}$ with an even integer $a$, the construction in this paper also applies to the case where $a$ is odd.
     \item[(iii)] The conditions $p\equiv 3\pmod{4}$ and $p\equiv 2\pmod{3}$ in (i) and (ii) of Theorem \ref{thm:1.2} are proposed in order to find explicit maximal elliptic curves.
         \end{itemize}
\end{rmk}

\subsection{Organization of the paper}
In Section 2, we introduce some preliminaries for this paper such as definition of locally repairable codes, elliptic curves over finite fields, field extension, algebraic geometry codes,  etc. In Section 3, we present explicit constructions of maximal/minimal elliptic curves with automorphism groups determined. In the last section, we give explicit constructions of  optimal locally repairable codes from elliptic curves given in Section 3. In addition, we will prove our main results, namely Theorems \ref{thm:1.1} and \ref{thm:1.2} in the last section.

\section{Preliminaries}\label{sec:2}
In this section, we present some preliminaries on locally repairable codes, elliptic curves over finite fields, field extension, algebraic geometry codes, etc.
\subsection{Locally repairable codes}\label{subsec:2.1}
Informally speaking, a block code is said with locality $r$ if  every coordinate of a given codeword can be recovered by accessing at most $r$ other coordinates of this codeword. The formal definition of a locally repairable code with locality $r$ is given as follows.

\begin{defn}
Let $C\subseteq \F_q^n$ be a $q$-ary block code of length $n$. For each $\Ga\in\F_q$ and $i\in \{1,2,\cdots, n\}$, define $C(i,\Ga):=\{\bc=(c_1,\dots,c_n)\in C\; | \; c_i=\Ga\}$. For a subset $I\subseteq \{1,2,\cdots, n\}\setminus \{i\}$, we denote by $C_{I}(i,\Ga)$ the projection of $C(i,\Ga)$ on $I$.
Then $C$ is called a locally repairable code with locality $r$ if, for every $i\in \{1,2,\cdots, n\}$, there exists a subset
$I_i\subseteq \{1,2,\cdots, n\}\setminus \{i\}$ with $|I_i|\le r$ such that  $C_{I_i}(i,\Ga)$ and $C_{I_i}(i,\Gb)$ are disjoint for any $\Ga\neq \Gb\in\F_q$.
\end{defn}
Apart from the usual parameters: length, rate and minimum distance,  the locality of a   locally repairable code plays a crucial role. In this paper, we always consider locally repairable codes that are linear over $\F_q$. Thus, a $q$-ary \LRC of length $n$, dimension $k$, minimum distance $d$ and locality $r$ is said to be an $[n,k,d]_q$-\LRC with locality $r$.

If we ignore the minimum distance of a $q$-ary locally repairable code, then there is a constraint on the rate \cite{GHSY12}, namely,
\begin{equation}\label{eq:x2}
\frac{k}{n} \le \frac{r}{r+1}.
\end{equation}

In this paper, the minimum distance of a \LRC is taken into consideration and we always refer an optimal code to a code achieving the bound \eqref{eq:x1}. For an $[n,k,d]$-linear code, $k$ information symbols can recover the whole codeword. Thus, the locality $r$ is usually upper bounded by $k$. If we allow $r=k$, i.e., there is no constraint on locality, then the bound \eqref{eq:x1} becomes the usual Singleton bound that shows constraint on $n,k$ and $d$ only. The other extreme case is that the locality $r$ is $1$. In this case, the \LRC is a repetition code by repeating each symbol twice and the bound \eqref{eq:x1} becomes $d(C)\le n-2k+2$ which shows  the Singleton bound for repetition codes.

\subsection{Elliptic curves over finite fields}

By a curve, we will always mean a projective, smooth and absolutely irreducible algebraic curve.
An elliptic curve $\fE$ over $\F_q$ is defined   by a  nonsingular Weierstrass equation
\begin{equation}\label{eq:3}
y^2+a_1xy+a_3y=x^2+a_2x^2+a_4x+a_6,
\end{equation}
where $a_i$ are elements of $\F_q$.
 An elliptic curve over $\F_q$ is also denoted by a pair $(\fE,O)$, where $\fE$ is the curve defined by the above Weierstrass equation \eqref{eq:3}, and $O$ is the common pole of $x$ and $y$ which is called the point at infinity of $\fE$. The genus of $\fE$ is $1$. Denote by $E:=\F_q(\fE)$ and $\fE(\F_q)$ the function field of $\fE$ and the set of $\F_q$-rational points, respectively. Then the function field   $E$ is given by $E=\F_q(x,y)$, where $x$ and $y$ satisfy the above Weierstrass equation \eqref{eq:3}. The set $\fE(\F_q)$ consists of $O$ and the solutions $(a,b)\in \F_q^2$ to the Weierstrass equation \eqref{eq:3} and forms an abelian group with $O$ as the zero element. There is one-to-one correspondence between $\F_q$-rational points of  the elliptic curve $\fE/\F_q$ and rational places of its function field $E/\F_q$. The rational point $(a,b)$ corresponds to the unique common zero of $x-a$ and $y-b$ and the point $O$ corresponds the common pole of $x$ and $y$.

Let $E/\F_q$ be an elliptic function field defined above. Let $\mathbb{P}_E$ be the set of all places of $E$ and $\mathbb{P}^1_E=\{P\in \mathbb{P}_E:\deg(P)=1\}$ be the set of rational places of $E$. Then we can identify $\mathbb{P}^1_E$ with $\fE(\F_q)$. This means that $\mathbb{P}^1_E$ is also an abelian group.
The divisor group of $E/\F_q$ is defined as the free abelian group generated by $\mathbb{P}_E$ and is denoted by Div$(E)$. The set of diviosrs of degree $0$ forms a subgroup of Div$(E)$, denoted by $\text{Div}^0(E)$.
Two divisors of $E$ are called equivalent if there exist $z\in E^*$ such that $A=B+(z)$, and we denote this by $A\sim B$.
The set of divisors $$\text{Princ}(E)=\{(x)=\sum_{P\in\PP_E}\nu_P(x): x\in E^*\}$$ is called the group of principal divisors of $E/\F_q$. It is a subgroup of $\text{Div}^0(E)$.
The factor group $$\text{Cl}^0(E)=\text{Div}^0(E)/\text{Princ}(E)$$ is called the zero divisor class group of $E/\F_q$.

Then there is a group isomorphism between  $\mathbb{P}^1_E=\fE(\F_q)$ and $\text{Cl}^0(E)$ given by.
$$\Phi:\begin{cases}\mathbb{P}^1_E\rightarrow \text{Cl}^0(E),\\ P\mapsto [P-O].\end{cases}$$
The group operation of $\mathbb{P}^1_E$ is denoted by $\oplus$  and the place $O$ is the zero element of the group $\mathbb{P}^1_E$. Thus, the following holds true  for any $P,Q,R\in \mathbb{P}_E^1$:
\begin{equation}\label{eq:4}P\oplus Q=R \Leftrightarrow P+Q\sim R+O.\end{equation}

The following lemma says that two distinct points can not be equivalent to each other due to isomorphism between  $\mathbb{P}^1_E=\fE(\F_q)$ and $\text{Cl}^0(E)$.
\begin{lemma}\label{lem:2.1}
Let $\fE/\F_q$ be an elliptic curve with function field $E=\F_q(\fE)$ and let $P,Q$ be two rational places of $E$. Then
$$P-Q=(z) \text{ for some } z\in E^* \text{ if and only if } P=Q.$$
\end{lemma}
The following lemma gives an upper bound on the size of $\fE(\F_q)$. It is called the Hasse-Weil bound \cite[Theorem 1.1 of Chapter 5]{Si86}.
\begin{lemma}\label{lem:2.2}
Let $\fE/\F_q$ be an elliptic curve defined over a finite field. Then
$$||\fE(\F_q)|-q-1|\le 2\sqrt{q}.$$
\end{lemma}
If the number $|\fE(\F_q)|$  of the rational places of $\fE$ attains the upper bound $1+q+2\sqrt{q}$ (or lower bound $1+q-2\sqrt{q}$, respectively), then $\fE$ is called a maximal elliptic curve over $\F_q$ (or a minimal elliptic curve over $\F_q$, respectively). Note that in this case $q$ must be  square of a prime power.

The zeta function of $\fE/\F_q$ is defined to be the following power series $$Z(\fE/\F_q;T)=\exp\left(\sum_{n=1}^{\infty} |\fE(\F_{q^n})| \frac{T^n}{n}\right),$$ where $\fE(\F_{q^n})$ stands for the set of $\F_{q^n}$-rational points of $\fE$. The zeta function of $\fE$ is a simple rational function (see \cite[Section 2 of Chapter 5]{Si86}).
\begin{lemma}\label{lem:2.3}
Let $\fE/\F_q$ be an elliptic curve. Then there is an integer $t\in \mathbb{Z}$ with $|t|\le 2\sqrt{q}$ such that
$$Z(\fE/\F_q;T)=\frac{1-tT+qT^2}{(1-T)(1-qT)}.$$
Furthermore, $t=q+1-|\fE(\F_q)|$ and $1-tT+qT^2=(1-\Ga T)(1-\Gb T)$ for some complex numbers $\Ga, \Gb$ with $|\Ga|=|\Gb|=\sqrt{q}.$
\end{lemma}
The polynomial $1-tT+qT^2$ is called the $L$-polynomial of $\fE/\F_q$ and denoted by $$L(\fE/\F_q, T)=1-tT+qT^2.$$
It is easy to see that $\fE/\F_q$ is maximal (or minimal, respectively) if and only if $t=-2\sqrt{q}$ and $L(\fE/\F_q, T)=1+2\sqrt{q}T+qT^2=(1+\sqrt{q}T)^2$ (or $t=2\sqrt{q}$ and $L(\fE/\F_q, T)=1-2\sqrt{q}T+qT^2=(1-\sqrt{q}T)^2$, respectively).

\subsection{Extension theory of function fields}
Let $E/\F_q$ be a function field with the full constant field $\F_q$.
Let $\mathbb{P}_E$ denote by the set of places of $E$ and let  $g(E)$ denote by the genus of $E$.
Let $G$ be a divisor of $E$.
The Riemann-Roch space $$\mathcal{L}(G)=\{z\in E^*: (z)\ge -G\}\cup \{0\}$$
is a finite dimensional vector space over $\F_q$ and its dimension $\dim_{\F_q}\mL(G)$ is at least $\deg(G)-g(E)+1$ from Riemann's theorem  \cite[Theorem 1.4.17]{St09}. If $E$ is an elliptic function field and $\deg(G)\ge 1$, then $\dim_{\F_q}\mL(G)=\deg(G)$.
Let $F$ be a subfield of $E$ with the same full constant field $\F_q$ such that $E/F$ is separable. Then
the Hurwitz genus  formula \cite[Theorem 3.4.13]{St09} yields
\begin{equation}\label{eq:5}2g(E)-2=[E:F](2g(F)-2)+\deg \text{Diff}(E/F),\end{equation}
where Diff$(E/F)$ stands for the different of $E/F$ \cite[Theorem 3.4.13]{St09}.

Let $\Aut(E/\F_q)$ be the automorphism group of the function field $E$ over $\F_q$, that is,
$$\Aut(E/\F_q)=\{\sigma:\;  \sigma \text{ is an } \F_q \text{-automorphism of } E\}.$$
Now we consider the group action of the automorphism group $\Aut(E/\F_q)$ on the set of places $\mathbb{P}_E$.
For any automorphism $\Gs\in \Aut(E/\F_q)$ and any place $P\in \mathbb{P}_E$, then $\Gs(P)=\{\Gs(z):z\in P\}$ is a place of $E$ as well.
Let $\mG$ be a subgroup of $\Aut(E/\F_q)$. The fixed subfield of $E$ with respect to $\mG$ is defined by
$$E^{\mG}=\{z\in E: \Gs(z)=z \text{ for all } \Gs\in \mG\}.$$
From the Galois theory, $E/E^{\mG}$ is a Galois extension with $\Gal(E/E^{\mG})=\mG.$
For any place $P\in \mathbb{P}_E$, the place $P\cap E^{\mG}$ is splitting completely in $E$ if and only if $\Gs(P)$ are pairwise distinct for all automorphisms $\Gs\in \mG$.

\begin{lemma} \label{lem:2.4}
Let $\mathfrak{E}/\F_q$ be an elliptic curve with function field $E=\F_q(\fE)$.  Let $F$ be the subfield of $E$ such that $E/F$ is a finite separable extension and there is a place $Q$ of $E$  with ramification index $e_Q>1$. Then $F$ is a rational function field.
\end{lemma}
\begin{proof} Let $d_Q$ be the different exponent of $Q$. Then $d_Q\ge e_Q-1\ge 1$ from Dedekind's Different Theorem \cite[Theorem 3.5.1]{St09}. By the Hurwitz genus formula \eqref{eq:5}, we have
\[0=2g(E)-2\ge [E:F](2g(F)-2)+d_Q\deg(Q).\]
This gives $g(F)\le 1-\frac{d_Q\deg(Q)}{2[E:F]}<1$. This forces that $g(F)=0$, i.e, $F$ is a rational function field.
\end{proof}

\subsection{Automorphism groups of elliptic curves} Consider the Weierstrass equation \eqref{eq:3}.
Let $b_2=a_1^2+4a_2$, $b_4=2a_4+a_1a_3$, $b_6=a_3^2+4a_6$ and $b_8=a_1^2a_6+4a_2a_6-a_1a_3a_4+a_2a_3^2-a_4^2.$ Then the quantity $\Delta$ defined by $$\Delta=-b_2^2b_8-8b_4^3-27b_6^2+9b_2b_4b_6$$
is called the discriminate of the Weierstrass equation \eqref{eq:3}. Then the Weierstrass equation \eqref{eq:3} is nonsingular if and only if $\Delta\neq 0$. In this case, it defines an elliptic curve $\fE$ over $\F_q$.
Furthermore, let $c_4=b_2^2-24b_4$ and $c_6=-b_2^3+36b_2b_4-216b_6$. Then the quantity $j(\fE)$ defined by $$j(\fE)=c_4^3/\Delta$$ is called the $j$-invariant of the elliptic curve $\fE$.
Two elliptic curves are isomorphic over the algebraic closure $\overline{\F}_q$ of $\F_q$ if and only if they both have the same $j$-invariant.
Under isomorphism, the Weierstrass equation \eqref{eq:3} of an elliptic curve can be simplified to some forms
\cite[Proposition 1.1 of Appendix A]{Si86}.

We denote by $\Aut(\fE)$ the set of  automorphisms of $\fE$ over the algebraic closure $\bar{\F}_q$, i.e., the automorphism group $\Aut(\fE/\bar{\F}_q)$. Then every automorphism $\Gs\in\Aut(\fE)$ fixes the zero point $O$.
 Let $\Aut(\fE/{\F}_q)$  the subgroup of $\Aut(\fE/\bar{\F}_q)$ in which every automorphism is defined over $\F_q$. The following result can be found in (see \cite[Theorem III.10.1]{Si86}).

\begin{lemma}\label{lem:2.5} Let $\mathfrak{E}/\F_q$ be an elliptic curve.  Then the order of $\Aut(\fE)$  divides $24$. More precisely speaking, the order of
$\Aut(\fE)$ is given by the following list:
\begin{itemize}
\item[{\rm (i)}] $|\Aut(\fE)|=2$  if  $j(\fE)\neq 0, 1728$;
\item[{\rm (ii)}] $|\Aut(\fE)|=4$ if $j(\fE)=1728$ and char$(\F_q)\neq 2,3$;
\item[{\rm (iii)}] $|\Aut(\fE)|=6$ if $j(\fE)=0$ and char$(\F_q)\neq 2,3$;
\item[{\rm (iv)}] $|\Aut(\fE)|=12$ if $j(\fE)=0=1728$ and char$(\F_q)=3$;
\item[{\rm (v)}] $|\Aut(\fE)|=24$ if $j(\fE)=0=1728$ and char$(\F_q)=2$.
\end{itemize}
\end{lemma}

Let $E$ be the function field $\F_q(\fE)$ of $\fE$ and denote by  $\Aut(E/\F_q)$ the automorphism group of $E$ fixed every element of $\F_q$. Let $F$ be the subfield of $E$ fixed by $\Aut(E/\F_q)$, i.e., $F=\{x\in E:\; \Gs(x)=x \text{ for all } \Gs\in \Aut(E/\F_q)\}$.  Then $\Aut(\fE/\F_q)=\Aut(E/\F_q)\cap \Aut(\fE)=\Gal(E/F)\cap\Aut(\fE)$ is the decomposition group of $\Gal(E/F)$ at the point $O$.

The following result shows that there are not many ramified places for an elliptic curve.
\begin{lemma} \label{lem:2.6}
Let $\mathfrak{E}/\F_q$ be an elliptic curve with function field $E=\F_q(\fE)$. Let $\mG$ be a subgroup of  $\Aut(\fE/\F_q)$ of $E$ at the point $O$ with order $|\mG|=r+1$ for a positive integer $r$.
Let $F=E^\mG$ be the fixed subfield of $E$ with respect to $\mG$. Then apart from the zero point $O$, there are at most other $r+2$ rational places of $E$ that are ramified in $E/F$.
\end{lemma}
\begin{proof}  $O$ is totally ramified in $E/F$ since $\mG$ is a subgroup of  $\Aut(\fE/\F_q)$. Assume that $P_1,\dots,P_m$ are $m$ ramified rational places of $E$ in $E/F$.  As $O$ has differnt exponent at least $r+1-1=r$ and each $P_i$ has different exponent at least $2-1=1$,   by the Hurwitz genus formula \eqref{eq:5}, we have
\[0=2g(E)-2\ge [E:F](2g(F)-2)+(r+1-1)+m(2-1)\ge -2(r+1)+r+m.\]
This gives the desired result.
\end{proof}
We now study certain specific elliptic curves and their automorphism groups.

\begin{lemma}\label{lem:2.7}
Let $q$ is an even power of $2$ and let $\fE$ be an elliptic curve over $\F_q$ defined by an equation $y^2+y=x^3+\Ga$ for some  $\Ga\in \F_q$. Then $|\Aut(\fE/\F_q)|=24$.
\end{lemma}
\begin{proof} By \cite[Proposition 1.1 of Appendix A]{Si86},  the $j$-invariant $j(\fE)$ is equal to $0$. An automorphism $\Gs\in \Aut(\fE/\F_q)$ is given by
\begin{equation}\label{eq:6}\Gs(x)=u^2x+s,\quad \Gs(y)=u^3y+u^2sx+t,\end{equation}
where $u,s,t\in\F_q$ satisfy $u^3=1$, $s^4+s=0$ and $t^2+t+s^6=0$ (see \cite[Proposition 1.2 of Appendix A]{Si86}). This gives $24$ solutions $(u,s,t)\in\F_q^*\times\F_q\times\F_q$. If we denote by $\Gs_{u,s,t}$ the automorphism given in \eqref{eq:6}, we have
\[\Aut(\fE/\F_q)=\{\Gs_{u,s,t}:\; u\in\F_4^*,\; s\in\F_4,\; t\in\F_4,\; t^2+t=s^3\}\]
and $|\Aut(\fE/\F_q)|=24$.
\end{proof}

\begin{lemma}\label{lem:2.8}
Let $q$ is an even power of $3$ and let $\fE$ be an elliptic curve over $\F_q$ defined by an equation $y^2=x^3+\Ga x$ for some  $\Ga\in \F_q^*$. Then \[|\Aut(\fE/\F_q)|=\left\{\begin{array}{ll} 4 &\mbox{if $-\Ga$ is a non-square in $\F_q^*$;}\\
12&\mbox{otherwise.}
\end{array}\right.\]
\end{lemma}
\begin{proof} By \cite[Proposition 1.1 of Appendix A]{Si86},  the $j$-invariant $j(\fE)$ is equal to $0$. An automorphism $\Gs\in \Aut(\fE/\F_q)$ is given by
\begin{equation}\label{eq:7}\Gs(x)=u^2x+s,\quad \Gs(y)=u^3y,\end{equation}
where $u,s\in\F_q$ satisfy $u^4=1$, $s^3+\Ga s=0$ (see \cite[Proposition 1.2 of Appendix A]{Si86}). This gives $4$ solutions (or $12$ solutions, respectively) $(u,s)\in\F_q^*\times\F_q^*$ if $-\Ga$ is a non-square (or square, respectively). If we denote by $\Gs_{u,s}$ the automorphism given in \eqref{eq:7}, we have
\[\Aut(\fE/\F_q)=\{\Gs_{u,s}:\; u\in\F_9^*,\; u^4=1,\; s\in\F_q^*,\; s^3+\Ga s=0\}.\]
The proof is completed.
\end{proof}

\begin{lemma}\label{lem:2.9}
Let $q$ is an even power of a prime $p$ with $p\neq 2,3 $ and $p\equiv 2\pmod{3}$ and let $\fE$ be an elliptic curve over $\F_q$ defined by an equation $y^2=x^3+\Ga$ for some $\Ga\in\F_q^*$. Then $|\Aut(\fE/\F_q)|=6$.
\end{lemma}
\begin{proof} In this case, the $j$-invariant $j(\fE)$ is equal to $0$. An automorphism $\Gs\in \Aut(\fE/\F_q)$ is given by
\begin{equation}\label{eq:xm2}\Gs(x)=u^2x,\quad \Gs(y)=u^3y,\end{equation}
where $u\in\F_q^*$ satisfies $u^6=1$ (see \cite[Proposition 1.2 of Appendix A]{Si86}). This gives $6$ solutions in $\F_q$ since $q\equiv 1\pmod{6}$ in this case. If we denote by $\Gs_u$ the automorphism given in \eqref{eq:xm2}, we have
\[\Aut(\fE/\F_q)=\{\Gs_u:\; u\in\F_{p^2}^*,\; u^6=1\}.\]
The proof is completed.
\end{proof}

\begin{lemma}\label{lem:2.10}
Let $q$ is an even power of a prime $p$ with $p\neq 2,3 $ and $p\equiv 3\pmod{4}$ and let $\fE$ be an elliptic curve over $\F_q$ defined by an equation $y^2=x^3+\Ga x$ for some $\Ga\in\F_q^*$. Then $|\Aut(\fE/\F_q)|=4$.
\end{lemma}
\begin{proof} In this case, the $j$-invariant $j(\fE)$ is equal to $1728$. An automorphism $\Gs\in \Aut(\fE/\F_q)$ is given by
\begin{equation}\label{eq:9}\Gs(x)=u^2x,\quad \Gs(y)=u^3y,\end{equation}
where $u\in\F_q^*$ satisfies $u^4=1$ (see \cite[Proposition 1.2 of Appendix A]{Si86}). This gives $4$ solutions in $\F_q$ since $q\equiv 1\pmod{4}$ in this case. If we denote by $\Gs_u$ the automorphism given in \eqref{eq:9}, we have
\[\Aut(\fE/\F_q)=\{\Gs_u:\; u\in\F_{p^2}^*, \; u^4=1\}.\]
The proof is completed.
\end{proof}

\subsection{Algebraic geometry codes}
For the construction of algebraic geometry codes,  the reader may refer to \cite{NX01,TV91,TVN90} for more details.
Let $F/\F_q$ be a function field with the full constant field $\F_q$. Let $\mP=\{P_1,\dots,P_n\}$ be a set of $n$ distinct rational places of $F$. For a divisor $G$ of $F$ with $0<\deg(G)<n$ and $\supp(G)\cap\mP=\emptyset$, the algebraic geometry code associated with $D$ and $G$ is defined to be
\begin{equation}\label{eq:10}
C(\mP,G):=\{(f(P_1),\dots,f(P_n)): \; f\in\mL(G)\}.
\end{equation}
Then $C(\mP,G)$ is an $[n,k,d]$-linear code with dimension $k=\dim_{\F_q}(G)$ and minimum distance $d\ge n-\deg(G)$.
If $V$ is a subspace of $\mL(G)$, we can define a subcode of $C(\mP,G)$ by
\begin{equation}\label{eq:11}
C(\mP,V):=\{(f(P_1),\dots,f(P_n)):\; f\in  V\}.
\end{equation}
Then the dimension of  $C(\mP,V)$ is the dimension of the vector space $V$ over $\F_q$ and the minimum distance of  $C(\mP,V)$ is still lower bounded by $n-\deg(G)$.

\section{Maximal and minimal elliptic curves}
In order to construct algebraic geometry codes with good parameters, we usually need function fields over finite fields with many rational places, especially maximal function fields. In this section, we provide explicit maximal elliptic curves which will be used in the latter section.

\begin{lemma}\label{lem:2.11}
Assume that $q$ is a square and $\fE$ is an elliptic curve over $\F_q$.
\begin{itemize}
\item[(i)] If  $\fE/\F_q$ is maximal, then  $\fE$ is maximal over $\F_{q^s}$ if and only if $s$ is odd. Furthermore, $\fE$ is minimal over $\F_{q^s}$ if and only if $s$ is even.
\item[(ii)]If  $\fE/\F_q$ is minimal, then  $\fE$ is minimal over $\F_{q^s}$ for all $s\ge 1$.
\end{itemize}
\end{lemma}
\begin{proof} Since $\fE/\F_q$ is maximal, its $L$-polynomial is $L(\fE/\F_q,T)=1+2\sqrt{q}T+qT=(1+\sqrt{q}T)^2$. Thus, the $L$-polynomial over $\F_{q^s}$ is $L(\fE/\F_{q^s},T)=(1-(-\sqrt{q})^sT)^2$ (see \cite[Section 5.1]{St09}). Since $\fE$ is maximal (or minimal, respectively) over $\F_{q^s}$ if and only is its zeta function is $(1+q^{s/2}T)^2$ (or $(1-q^{s/2}T)^2$, respectively). The desired result of part (i) follows.

The proof of part (ii) is similar and we skip the detail.
\end{proof}

An elliptic curve $\fE/\F_q$ is supersingular if its number of $\F_q$-rational points is equal $1+q+t$ for an integer $t$ divisible by the characteristic of $\F_q$. Two elliptic curves over $\F_q$ are said isogenous if they have the same number of $\F_q$-rational points.
\begin{lemma}\label{lem:2.12}(see \cite{Wa69})
The isogeny classes of elliptic curves over $\F_q$ for $q=p^a$ are in one-to-one correspondence with the rational integers $t$ having $|t|\le 2\sqrt{q}$
and satisfying some one of the following conditions:
\begin{itemize}
\item[(i)] $(t,p)=1$;
\item[(ii)] If $a$ is even: $t=\pm 2\sqrt{q}$;
\item[(iii)] If $a$ is even and $p\not \equiv 1\pmod{3}$: $t=\pm \sqrt{q}$;
\item[(iv)] If $a$ is odd and $p=2$ or $3$: $t=\pm p^{\frac{a+1}{2}}$;
\item[(v)] If either (i) $a$ is odd or (ii) $a$ is even and $p\not \equiv 1\pmod{4}: t=0.$
\end{itemize}
The first of these is not supersingular; the rest are supersingular.
\end{lemma}

By making use of Lemma \ref{lem:2.12} or counting points directly, we provide some examples of maximal elliptic curves with automorphism group determined in this section.

\begin{lemma}\label{lem:2.13} For any even $a$, there exists a maximal elliptic curve $\fE$ over $\F_{2^a}$ defined by an equation $y^2+y=x^3+\Gg$ for some $\Gg\in\F_{2^a}$ such that it has an automorphism group of size $24$, i.e., $|\Aut(\fE/\F_{p^a})|=24$.
\end{lemma}
\begin{proof} Let $q=2^a$. First assume that $a\equiv 2\pmod{4}$.  Consider the elliptic curve $\fE$ defined by $y^2+y=x^3$. It is straightforward to verify that it has $9=1+4+2\sqrt{4}$ rational points over $\F_4$ and hence it is maximal over $\F_4$. Then $\fE$ is  maximal  over $\F_{2^a}$ for $a\equiv 2\pmod{4}$ by Lemma \ref{lem:2.11}(i).

Now let $a\equiv 0\pmod{4}$. By Lemma \ref{lem:2.11}(i), the elliptic curve $\fE$ is minimal over $\F_{2^4}$. Then  $\fE$ is minimal over $\F_q$ for  $a\equiv 0\pmod{4}$ by Lemma \ref{lem:2.6}(ii). We claim that the twisted curve $\fE^{\prime}$  defined by $y^2+y+\Gg=x^3$ for some $\Gg \in \F_{q}\setminus \{\Gb^2+\Gb:\Gb\in \F_q\}$ is maximal. Indeed, let $N$ and $N'$ denote the number of $\F_q$-rational points of $\fE$ and $\fE'$, respectively.
For any $\Ga\in \F_q$, it is easy to see that one of $y^2+y=\Ga^3$ and $y^2+y+\Gg=\Ga^3$ has no solutions and other has two solutions. This implies that $N+N'=2+2q$, where the point $O$ at infinity is counted twice. Since $\fE$ is minimal, we have $N=1+q-2\sqrt{q}$. Hence, $N'=1+q+2\sqrt{q}$, i.e., $\fE'$ is maximal over $\F_q$.

By Lemma \ref{lem:2.7},  we have $|\Aut(\fE/\F_q)|=|\Aut(\fE'/\F_q)|=24$.
\end{proof}

\begin{lemma}\label{lem:2.14} Let $p\equiv 3\pmod{4}$ be a prime. Then for any even $a$, there  exists a maximal elliptic curve $\fE$ over $\F_{p^a}$ defined by an equation $y^2=x^3+ \theta^2 x$ for some $\theta\in\F_{p^a}^*$ such that
\[|\Aut(\fE/\F_{p^a})|=\left\{\begin{array}{ll} 4 &\mbox{if $p\neq 3$;}\\
12&\mbox{if $p=3$.}
\end{array}\right.\]
\end{lemma}
\begin{proof} Let $q=p^a$. By Example 4.5 of \cite[Chpater V]{Si86}, the elliptic curve $\fE/\F_p$ defined by $y^2=x^3+x$ is supersingular. If $p\neq 3$, then by Lemma \ref{lem:2.12}(v), $\fE$ has $1+p$ rational points over $\F_p$. This implies that the $L$-polynomial $L(\fE/\F_p,T)$ is equal to $1+pT^2=(1-i\sqrt{p}T)(1+i\sqrt{p}T)$, where $i$ is the imaginary unit. Thus, the $L$-polynomial $L(\fE/\F_{p^2},T)$ is equal to $(1-(i\sqrt{p})^2T)(1-(-i\sqrt{p})^2T)=(1+pT)^2$, i.e, $\fE$ is maximal over $\F_{p^2}$. If $p=3$, then it is straightforward to verify that the curve  $\fE$ is maximal over $\F_9$. Thus, for any prime $p$ with $p\equiv 3\pmod{4}$, the elliptic curve $\fE/\F_{p^2}$ defined by $y^2=x^3+x$ is maximal.

By Lemma \ref{lem:2.11}(i), the elliptic curve $\fE/\F_{p^a}$ defined by $y^2=x^3+x$ is maximal for $a\equiv 2\pmod{4}$. Furthermore, since $-1$ is a square in $\F_q$, by Lemma \ref{lem:2.8},  we have $|\Aut(\fE/\F_q)|=12$ for $p=3$. For $p\neq 3$, it follows from Lemma \ref{lem:2.10} that  $|\Aut(\fE/\F_q)|=4$.

If $a\equiv 0\pmod{4}$, then by Lemma \ref{lem:2.11}(i) the curve $\fE$ is minimal over $\F_{p^a}$. By using the similar way as in the proof of Lemma \ref{lem:2.13}, the twist elliptic curve $\fE_1$  given by $\theta y^2=x^3+x$ for any non-square $\theta \in \F_{q}$ is maximal. By multiplying $\theta^3$ on the both side of the equation and then substituting $\theta x$ by $x$, $\theta^2 y$ by $y$, we get the  elliptic curve $\fE_2$  defined by $y^2=x^3+\theta^2 x$ that is isomorphic to $\fE_1$ over $\F_q$. Hence, it is still maximal over $\F_q$.  For $p\neq 3$, it follows from Lemma \ref{lem:2.10} that  $|\Aut(\fE_2/\F_q)|=4$. Furthermore, since $-1$ and $\theta^2$ are squares in $\F_q$, by Lemma \ref{lem:2.8},  we have $|\Aut(\fE_2/\F_q)|=12$ for $p=3$.
\end{proof}

\begin{lemma}\label{lem:2.15} Let $p\neq 2$ and $p\equiv 2\pmod{3}$ be an odd prime. Then for any even $a$, there  exists a maximal elliptic curve over $\F_{p^a}$ defined by an equation $y^2=x^3+ \theta^3$ for some $\theta\in\F_{p^a}^*$ such that it has an automorphism group of size $6$, i.e., $|\Aut(\fE/\F_{p^a})|=6$.
\end{lemma}
\begin{proof} Let $q=p^a$. By Example 4.4 of \cite[Chpater V]{Si86}, the elliptic curve $\fE/\F_p$ defined by $y^2=x^3+1$ is supersingular. In the same way as we proved in Lemma \ref{lem:2.14}, one can show that $\fE$ is maximal over $\F_{p^2}$.  Thus, for any prime $p$ with $p\equiv 3\pmod{4}$, the elliptic curve $\fE/\F_{p^2}$ defined by $y^2=x^3+1$ is maximal.

By Lemma \ref{lem:2.11}(i), the elliptic curve $\fE/\F_{p^a}$ defined by $y^2=x^3+1$ is maximal for $a\equiv 2\pmod{4}$. It follows from Lemma \ref{lem:2.9} that  $|\Aut(\fE/\F_q)|=6$.

If $a\equiv 0\pmod{4}$, then by Lemma \ref{lem:2.11}(i) the curve $\fE$ is minimal over $\F_{p^a}$. By using the similar arguments as in the proof of Lemma \ref{lem:2.13}, one can show that the twist elliptic curve $\fE_1$  given by $\theta y^2=x^3+1$ for any non-square $\theta \in \F_{q}^*$ is maximal. By multiplying $\theta^3$ on the both side of the equation and then substituting $\theta x$ by $x$, $\theta^2 y$ by $y$, we get the  elliptic curve $\fE_2$  defined by $y^2=x^3+\theta^3 $ that is isomorphic to $\fE_1$ over $\F_q$. Hence, it is still maximal over $\F_q$.  It follows from Lemma \ref{lem:2.9} that  $|\Aut(\fE_2/\F_q)|=6$.
\end{proof}

\section{Construction of locally repairable codes via elliptic curves}\label{sec:3}

\subsection{A general construction via automorphism groups}
Let $\fE/\F_q$ be an elliptic curve with the  function field $E$. Let $\Aut(E/\F_q)$ be the automorphism group of $E$ over $\F_q$.
Let $\mG$ be a subgroup of $\Aut(E/\F_q)$ of order $r+1$ and let $E^{\mG}$ be the fixed subfield of $E$ with respect to $\mG$.
Denote by $F$ the fixed subfield $E^{\mG}$. Then $E/F$ is a Galois extension with Galois group $\Gal(E/F)=\mG.$

Assume that there exist
$\ell$ rational places $Q_1,\cdots, Q_\ell$ of $F$ which are all splitting completely in $E/F$.
Let $P_{i,1}, P_{i,2},\cdots, P_{i,r+1}$ be the $r+1$ rational places of $E$ lying over $Q_i$ for each $1\le i \le \ell$.
Put $\mP=\{P_{i,j}: 1\le i \le \ell, 1\le j\le r+1\}$. Then the cardinality of $\mP$ is $\ell(r+1)$.

Choose a divisor $G$ of $F$ such that $\text{supp}(G) \cap \{Q_1,\cdots,Q_m\}=\emptyset$.
The Riemann-Roch space $\mathcal{L}(G)=\{f\in F^*: (f)\ge -G\}\cup \{0\}$
 is a finite dimensional vector space over $\F_q$ with dimension $\dim_{\F_q}(G)\ge \deg(G)-g(F)+1$, where $g(F)$ is the genus of $F$.
Let $\{z_1,\cdots, z_t\}$ be a basis of the Riemann-Roch space $\mathcal{L}(G)$ over $\F_q$.
Choose elements $w_i\in E$ such that $w_0=1,w_1,\cdots,w_{r-1}$ are linearly independent over $F$ and $\nu_{P_{i,j}}(w_l)\ge 0$ for all $1\le i\le \ell$, $1\le j\le r+1$ and $0\le l\le r-1$.
Consider the set of functions
\begin{equation}\label{eq:12} V:=\left\{\sum_{j=1}^t a_{0j}z_j+\sum_{i=1}^{r-1} \left(\sum_{j=1}^{t-1} a_{ij}z_j\right)w_i\in E:\; a_{ij}\in \F_q \right\}.\end{equation}

\begin{prop}\label{prop:3.1}
(see \cite{BHHMV16})
Let $i$ be an integer between $1$ and $\ell$, and suppose that every $r\times r$ submatrix of the matrix
$$M=\left(\begin{array}{cccc}w_0(P_{i,1}) & w_1(P_{i,1}) & \cdots & w_{r-1}(P_{i,1}) \\w_0(P_{i,2}) & w_1(P_{i,2})  & \cdots &w_{r-1}(P_{i,2})  \\\vdots & \vdots & \ddots & \vdots \\w_0(P_{i,r+1}) & w_1(P_{i,r+1})  & \cdots &w_{r-1}(P_{i,r+1})\end{array}\right)$$
is invertible. Then the value of $f\in V$ at any place in the set $\{P_{i,1}, P_{i,2},\cdots, P_{i,r+1}\}$ can be recovered from the values of $f$ at the other $r$ places.
\end{prop}


\begin{prop}\label{prop:3.2}
Let $\mP$ and $V$ be defined as above and satisfy the assumption of Proposition {\rm \ref{prop:3.1}}. If $V$ is contained in $\mL(D)$ for a divisor $D$ of $E$ with $\deg(D)< \ell (r+1)$ and $\supp(D)\cap\{P_{i,1},\dots,P_{i,r+1}\}_{i=1}^{\ell}=\emptyset$, then the algebraic geometry code
$$C(\mP,V)=\{(f(P))_{P\in \mP}: f\in V\}$$ is a $q$-ary $[n,k,d]$-locally repairable code with locality $r$, length $n=\ell(r+1)$, dimension $k=rt-(r-1)$ and minimum distance $d\ge n-\deg(D)$.
\end{prop}
\begin{proof}
First, it is easy to see that the dimension of $V$ over $\F_q$ is $rt-(r-1)$,
since $w_0,w_1,\cdots, w_{r-1}$ are linearly independent over $F$ and $\{z_1,\cdots, z_t\}$ is a basis of the Riemann-Roch space $\mathcal{L}(G)$ over $\F_q$.
Every nonzero function $f\in V\subseteq\mL(D)$ has at most $\deg(D)$ zeros among $\{P_{i,1},\dots,P_{i,r+1}\}_{1\le i\le \ell}$. Hence, the minimum distance of $C(\mP,V)$ is lower bounded by $d\ge n-\deg(D)$. Under the assumption that $\deg(D)<\ell(r+1)=n$, we have $d\ge 1$. Hence, the dimension of  $C(\mP,V)$ is $k=\dim_{\F_q}V=rt-(r-1)$.
The locality property follows from  Proposition \ref{prop:3.1}.
\end{proof}

By considering  subgroups of the automorphism group $\Aut(\fE)$, we can choose a space $V$ and some rational points such that the assumption of Proposition \ref{prop:3.1} is satisfied.

\begin{prop}\label{prop:3.3}
Let $\fE$ be an elliptic curve defined by the Weierstrass equation \mbox{\eqref{eq:3}}.
Let $E$ be the function field $\F_q(\fE)$.
Let $\mG$ be a subgroup of $\Aut(\fE/\F_q)$  with order $|\mG|=r+1=2s$ for  positive integers $s\ge 2$ and $r<q$. Furthermore, we assume that  the set $\{\Gs(x):\; \Gs\in\mG\}$ has size $s$.
Let $F=E^\mG$ be the fixed subfield of $E$ with respect to $\mG$. Then
\begin{itemize}
\item[{\rm (i)}] there exists an element  $z\in E$ satisfying that $F=\F_q(z)$ and $\deg(z)_\infty=r+1$; and elements   $w_0=1,w_1,\cdots, w_{r-1}$ of $E$ that are linearly independent over $F$;

\item[{\rm (ii)}] let  $\{P_{i,1},P_{i,2},\cdots,P_{i,r+1}\}$ be the pairwise distinct rational places lying over the same place of $F$ for each $1\le i\le \ell$, such that
$\{P_{i,1},P_{i,2},\cdots,P_{i,r+1}\}_{i=1}^l \cap \supp(D)=\emptyset.$
Then every $r\times r$ submatrix of the matrix
$$M=\left(\begin{array}{cccc}1 & w_1(P_{i,1}) & \cdots & w_{r-1}(P_{i,1}) \\1 & w_1(P_{i,2})  & \cdots &w_{r-1}(P_{i,2})  \\\vdots & \vdots & \ddots & \vdots \\1& w_1(P_{i,r+1})  & \cdots &w_{r-1}(P_{i,r+1})\end{array}\right)$$
is invertible for all $1\le i\le \ell$.
\end{itemize}
\end{prop}

\begin{proof} 
Let $\Gs_1=1, \Gs_2, \cdots, \Gs_s$ be the automorphisms of $\mG$ with the pairwise distinct $\Gs_i(x)$.
Let $\Gs_{s+j}$ be the automorphisms of $\mG$ with  $\Gs_j(x)=\Gs_{s+j}(x)$ for $1\le j\le s$.
Let $P=(a,b)$ be a rational place of $E$ such that $\Gs(P)$ are pairwise distinct for all $\Gs\in \mG$, i.e., $P\cap F$ is splitting completely in $E$.

Put $z=\prod_{i=1}^s\frac{1}{\Gs_i^{-1}(x)-a}$. Then $\Gs(z)=z\in F$ and the principal divisor of $z$ is
$$(z)=(r+1)O-P_1-P_2-\cdots-P_{r+1},$$
where $\{P_1,P_2,\dots,P_{r+1}\}=\{\Gs(P):\; \Gs\in\mG\}$.
Hence, we obtain $\deg(z)_{0}=r+1=[E:F]$ and hence $F=\F_q(z)$ by \cite[Theorem 1.4.11]{St09}.
As $\ell(P_1+P_2)=2$ from Riemann's Theorem, there exists an element $w_1\in \mL(P_1+P_2)\setminus \F_q$ such that $(w_1)_{\infty}=P_1+P_2$.
For each $2\le i\le r-1$, the set $\cup_{j=1}^{i+1}\mL(\sum_{u=1}^{i+1}P_u-P_j)$ has size at most $(i+1)q^{i}$ that is less than $q^{i+1}=|\mL(\sum_{u=1}^{i+1}P_u)|$. This implies that there exists  an element $w_i\in E$ such that $(w_i)_{\infty}=P_1+P_2+\cdots+P_{i+1}$. 
We claim that $w_0=1, w_1, \cdots, w_{r-1}$ are linearly independent over $\F_q(z)$. Suppose that $w_0, w_1, \cdots, w_{r-1}$ are linearly dependent over $\F_q(z)$, i.e., there exist functions $a_0(z),  a_1(z),\dots,a_{r-1}(z)$ that are not all zero such that $\sum_{i=0}^{r-1}a_i(z)w_i=0$. By multiplying a common nonzero polynomial in $\F_q[z]$, we may assume that every $a_i(z)$ is a polynomial of $\F_q[z]$. Let $0\le s\le r-1$ be the largest integer such that $a_s(z)$ has the largest degree, i.e., $\deg(a_s(z))=\max\{\deg(a_i(z))\}_{i=0}^{r-1}$ and $\deg(a_s(z))>\max\{\deg(a_{s+1}(z)),\dots,\deg(a_{r-1}(z))\}$, where degree of the zero polynomial is defined to be $-\infty$. Thus, for $i<s$, we have \begin{eqnarray*}\nu_{P_{s+1}}(a_i(z)w_i)&=&-\deg((a_i(z))+\nu_{P_{s+1}}(w_i)=-\deg((a_i(z))\\
&>& -\deg((a_i(z))-1\ge  -\deg((a_s(z))-1=\nu_{P_{s+1}}(a_s(z)w_s).\end{eqnarray*}
For $i>s$, we have \begin{eqnarray*}\nu_{P_{s+1}}(a_i(z)w_i)&=&-\deg((a_i(z))+\nu_{P_{s+1}}(w_i)> -\deg((a_s(z))+\nu_{P_{s+1}}(w_i) \\
&=& -\deg((a_s(z))+\nu_{P_{s+1}}(w_s)=\nu_{P_{s+1}}(a_s(z)w_s).\end{eqnarray*}
This implies that $\nu_{P_{s+1}}(-a_s(z)w_s)<\nu_{P_{s+1}}(\sum_{0\le i\le r-1,i\neq s}a_i(z)w_i)$ by the Strictly Triangle Inequality \cite[Lemma 1.1.11]{St09}.
It is a contradiction since $-a_s(z)w_s=\sum_{0\le i\le r-1,i\neq s}a_i(z)w_i$.

If the pairwise distinct places $P_{i,1},\cdots,P_{i,r+1}$ lie over the same rational place $z-\Gb_i$ of $F$ for some $\Gb_i\in\F_q$, we claim
every $r\times r$ submatrix of the matrix
$$M=\left(\begin{array}{cccc}1 & w_1(P_{i,1}) & \cdots & w_{r-1}(P_{i,1}) \\1 & w_1(P_{i,2})  & \cdots &w_{r-1}(P_{i,2})  \\\vdots & \vdots & \ddots & \vdots \\1& w_1(P_{i,r+1})  & \cdots &w_{r-1}(P_{i,r+1})\end{array}\right)$$
is invertible. Without loss of generality, we may consider the first $r$ rows. Suppose that
$$\text{det}\left(\begin{array}{cccc}1 & w_1(P_{i,1}) & \cdots & w_{r-1}(P_{i,1}) \\1 & w_1(P_{i,2})  & \cdots &w_{r-1}(P_{i,2})  \\\vdots & \vdots & \ddots & \vdots \\1& w_1(P_{i,r})  & \cdots &w_{r-1}(P_{i,r})\end{array}\right)=0.$$
Then there exists $(c_0,\cdots,c_{r-1})\in \F_q^r\setminus\{\bo\}$  such that
$$\left(\begin{array}{cccc}1 & w_1(P_{i,1}) & \cdots & w_{r-1}(P_{i,1}) \\1 & w_1(P_{i,2})  & \cdots &w_{r-1}(P_{i,2})  \\\vdots & \vdots & \ddots & \vdots \\1& w_1(P_{i,r})  & \cdots &w_{r-1}(P_{i,r})\end{array}\right)
\left(\begin{array}{c} c_0 \\  c_1\\  \vdots \\ c_{r-1}\end{array} \right)=0.$$
Then we have $(c_0+c_1w_1+\cdots+c_{r-1}w_{r-1})(P_{i,j})=0$ for all $1\le j\le r$ and hence $c_0+c_1w_1+\cdots+c_{r-1}w_{r-1}\in \mL(P_1+\cdots+P_r-P_{i,1}-\cdots-P_{i,r})$. Thus, the principal divisor of  $c_0+c_1w_1+\cdots+c_{r-1}w_{r-1}$ is
$$(c_0+c_1w_1+\cdots+c_{r-1}w_{r-1})=\sum_{j=1}^r P_{i,j}-\sum_{j=1}^r P_{j}.$$
As the places $P_{i,1},\cdots,P_{i,r}$ lie over the same rational place of $z-\Gb_i$, the $x$-coordinate of $P_{i,j}$ for $1\le j\le r$ are the roots of
$$\prod_{j=1}^s\frac{1}{\Gs_j^{-1}(x)-a}-\Gb_i=0.$$
Let $\Ga_1,\Ga_2,\cdots, \Ga_s\in\F_q$ be pairwise distinct roots of the above equation.
After rearranging order of the places $P_{i,j}$ for $1\le j\le r$, we may assume that
$$(x-\Ga_j)=P_{i,2j-1}+P_{i,2j}-2O \text{ for } 1\le j\le s-1.$$
Since we have $P_j+P_{s+j}-2O=(\Gs_j^{-1}(x)-a)$ for $1\le j\le s-1$, then
$$P_{i,r}-P_s=\Big{(}(c_0+c_1w_1+\cdots+c_{r-1}w_{r-1})\prod_{j=1}^{s-1}\frac{(\Gs_j^{-1}(x)-a)}{x-\Ga_j}\Big{)}.$$
This is a contradiction by Lemma \ref{lem:2.1}.
\end{proof}

With the above preparation, we can now state a result on locally repairable codes from elliptic curves.
\begin{prop}\label{prop:3.4}
Let $\fE$ be an elliptic curve with $N$ rational points defined by the Weierstrass equation \mbox{\eqref{eq:3}}.
Let $E$ be the function field $\F_q(\fE)$.
Let $\mG$ be a subgroup of $\Aut(\fE/\F_q)$  with order $|\mG|=r+1=2s$ for  positive integers $s\ge 2$ and $r<q$. Furthermore, we assume that  the set $\{\Gs(x):\; \Gs\in\mG\}$ has size $s$.
 Then there exists an optimal $q$-ary $[n=\ell(r+1),k=rt-r+1, d=n-(t-1)(r+1)]$ locally repairable code with locality $r$ for any $1\le \ell \le \left\lfloor\frac{N-2r-4}{r+1}\right\rfloor$ and $1\le t\le\ell$.
\end{prop}
\begin{proof} Let $F=E^\mG$ be the fixed subfield of $E$ with respect to $\mG$. Let $z\in F$ and $w_0,w_1,\dots,w_{r-1}\in E$ be elements given in Proposition \ref{prop:3.3}. By Lemma \ref{lem:2.6}, there are at most $r+3$ rational places (including $O$) in total that are ramified in $E/F$. Except for all ramified places, all other rational places of $E$ are splitting completely in $E/F$. By our assumption, we can find $\ell$ sets $\{P_{i,1},\dots,P_{i,r+1}\}_{i=1}^{\ell}$ that do not intersect with ramified points and $r+1$ pole of $z$. Put $z_i=z^{i-1}$ for $i=1,2,\dots,t$ and consider the set $V$ of functions given in \eqref{eq:12}.
Then $V$ is a subspace of $\mL((t-1)(P_1+\cdots+P_{r+1}))$, where $P_1,P_2,\dots,P_{r+1}$ are $r+1$ pole places of $z$ given in the proof of Proposition \ref{prop:3.3}.
By Proposition \ref{prop:3.2},
 the algebraic geometry code
$C(\mP,V)$ is an $[n=\ell(r+1),k=rt-r+1,d\ge n-(t-1)(r+1)]$ locally repairable codes with locality $r$. On the other hand, by the Singleton-type bound \eqref{eq:x1},
\[d\le n-k- \left\lceil \frac kr\right\rceil+2=n-rt+r-1-\left\lceil \frac {tr-r+1}r\right\rceil+2=n-(t-1)(r+1).\]
This implies that $C(\mP,V)$ is optimal and the desired result follows.
\end{proof}
\begin{rmk}\label{rmk:2}
By modifying algebraic geometry codes, we can include the poles of $z$ in the set of evaluation points (see the details in Sections 2 and 3 of \cite{JMX17}). Thus, $\ell$ in Proposition \ref{prop:3.4} can take values up to $\left\lfloor\frac{N-r-3}{r+1}\right\rfloor$ instead of $\left\lfloor\frac{N-2r-4}{r+1}\right\rfloor$.
\end{rmk}
Since we consider  subgroups of $\Aut(\fE)$, the locality $r$ can only take one of values $1,2,3,5,7,11,23$. Locality $1$ is a trivial case. Let us start with locality $2$.

\subsection{Locality $r=2$}
As Proposition \ref{prop:3.4} requires that $r+1$ be an even number, i.e., $r$ is odd, we need a different construction for the case where the locality $r=2$.
\begin{prop}\label{prop:3.5}
Let $\mathfrak{E}/\F_q$ be an elliptic curve with function field $E=\F_q(\fE)$. Assume that $\Aut(\fE/\F_q)$ contains a subgroup $\mG$ of order $3$. Let $F$ be the subfield of $E$ with $\Gal(E/F)=\mG$. Assume that there are $\ell$ rational places of $F$ that split completely in $E/F$. Then for any $t$ with $0\le t<\ell$, there exists a $q$-ary $[n=3\ell,k=2t+1,d=n-3t]$ locally repairable code with locality $2$.
\end{prop}
\begin{proof} Let $O$ be the zero element of $\fE$. Then $O$ is totally ramified in $E/F$ and hence $F$ is a rational function field by Lemma \ref{lem:2.4}. Let $O'$ be the unique place of $F$ that lies under $O$. Choose $z$ in $F$ such that $(z)_\infty=O'$ as a divisor of $F$ (this is possible as $F$ is a rational function field). Then $(z)_\infty=3O$ as a divisor of $E$.  Choose $x\in E$ such that $(x)_\infty=2O$.

Consider the  $\F_q$-space $V_t$ defined by
\begin{equation}
V_t:=\{f_0(z)+f_1(z)x:\; f_i(z)\in\F_q[z] \text{ for } i=0,1;\;  \deg(f_0(z))\le t;\; \deg(f_1(z))\le t-1\}.
\end{equation}
It is clear that $V_t$ is a subspace of the Riemann-Roch space $\mL(3tO)$. We claim that $\dim_{\F_q}V_t=2t+1$. Suppose that $f_0(z)+f_1(z)x$ is the zero function and one of $f_0(z)$ and $f_1(z)$ is not the zero polynomial. Then both polynomials must be nonzero. Hence $f_0(z)=-f_1(z)x$ and $-3\deg(f_0)=\nu_{O}(f_0)=\nu_O(-f_1(z)x)=-3\deg(f_1)-2$. This is a contradiction and it implies that $\dim_{\F_q}V_t=2t+1$.

Let  $\{P_{i1},P_{i2},P_{i3}\}_{i=1}^{\ell}$ be $n$ rational places of $E$ such that, for each $i$,  the three places $P_{i1},P_{i2},P_{i3}$ lie over the same rational place of $F$.
Define the algebraic geometry code
\begin{equation} C_t:=\{(f(P_{i1}),f(P_{i2}),f(P_{i3}))_{i=1}^{\ell}:\; f\in V_t\}
\end{equation}
Then the dimension of the code $C_t$ is $k=2t+1$ and the minimum distance $d$ is at least $n-3t>0$ as $f\in \mL(3tO)$.
The Singleton-type bound \eqref{eq:x1} of locally repairable codes shows that $d\le n-k-\left\lceil \frac k2\right\rceil+2=n-(2t+1)-(t+1)+2=n-3t$.
Hence, the code $C_t$ is an optimal $q$-ary $[3\ell, 2t+1,3\ell-3t]$-locally repairable code with locality $2$.
\end{proof}

{\bf Proof of Theorem \ref{thm:1.1}:}
\begin{proof} If $p=2$, then by Lemma \ref{lem:2.13}, there exists a maximal elliptic curve $\fE$ defined by (i) $y^2+y=x^3$ for $a\equiv 2\pmod{4}$;  or (ii) $y^2+y=x^3+\Gg$ for $a\equiv 0\pmod{4}$ and some $\Gg\in\F_q\setminus\{\Ga^2+\Ga:\; \Ga\in\F_q\}$  with $|\Aut(\fE/\F_q)|=24$. Furthermore, the automorphism $\Gs_{u,0,0}$ given by $\Gs_{u,0,0}(x)=u^2x$ and $\Gs_{u,0,0}(y)= u^3y$ for a $3$rd primitive root $u$ of unity in $\F_q$ generates a cyclic group $\mG$ of order $3$. Let  $F=E^{\mG}$, where $E=\F_q(\fE)$ is the function field of $\fE$. The zero point $O$ is a totally ramified point with respect to the extension $E/F$.

In addition, for $a\equiv 2\pmod{4}$, the point $(0,0)$ and $(0,1)$ are also  totally ramified since $\Gs_{u,0,0}(0,0)=(0,0)$ and $\Gs_{u,0,0}(0,1)=(0,1)$. Thus, the rest of $q+2\sqrt{q}-2$ rational points split completely. By Proposition \ref{prop:3.5}, there exists an optimal $q$-ary $[n=3\ell, k=2t+1,d=n-3t]$ locally repairable code with locality $2$ for any $0\le t<\ell\le \frac{q+2\sqrt{q}-2}3$. Note that in this case we have  $\frac{q+2\sqrt{q}-2}3=\left\lfloor\frac{q+2\sqrt{q}}3\right\rfloor$.

For $a\equiv 0\pmod{4}$, the curve is given by $y^2+y=x^3+\Gg$. In this case, the point $(0,\Ga)$ and $(0,\Gb)$ are   totally ramified for some $\Ga,\Gb\in\bar{\F}_{q}\setminus\F_q$. Thus, there are $q+2\sqrt{q}$ rational points that split completely in $E/F$. By Proposition \ref{prop:3.5}, there exists an optimal $q$-ary $[n=3\ell, k=2t+1,d=n-3t]$ locally repairable code with locality $2$ for any $1\le t<\ell\le \frac{q+2\sqrt{q}}3$. Note that in this case we have  $\frac{q+2\sqrt{q}}3=\left\lfloor\frac{q+2\sqrt{q}}3\right\rfloor$.

If $p=3$, then by Lemma \ref{lem:2.14}, there exists a maximal curve elliptic curve $\fE$ defined by  $y^2=x^3+\theta^2 x$ for some $\theta\in\F_q^*$  with $|\Aut(\fE/\F_q)|=12$. In particular, $\mG:=\{\Gs_{1,s}:\; s\in\F_q,\; s^3+\theta^2 s=0\}$ is a group of order $3$, where $\Gs_{1,s}$ given by $\Gs_{1,s}(x)=x+s$ and $\Gs_{1,s}(y)= y$. Let  $F=E^{\mG}$, where $E=\F_q(\fE)$ is the function field of $\fE$. Then the zero point $O$ is the unique totally ramified point with respect to the extension $E/F$. Thus, the rest of $q+2\sqrt{q}$ splits completely.  By Proposition \ref{prop:3.5}, there exists an optimal $q$-ary $[n=3\ell, k=2t+1,d=n-3t]$ locally repairable code with locality $2$ for any $0\le t<\ell\le \frac{q+2\sqrt{q}}3$. Note that in this case we have  $\frac{q+2\sqrt{q}}3=\left\lfloor\frac{q+2\sqrt{q}}3\right\rfloor$.

Finally, assume that  $p\neq 2,3$. By Lemma \ref{lem:2.15}, the elliptic curve $\fE/\F_q$ defined by $y^2=x^3+\theta^2$ for some $\theta\in\F_q^*$ is maximal with $|\Aut(\fE/\F_q)|=6$. Then the automorphism $\Gs_{u}$ given by $\Gs_{u}(x)=u^2x$ and $\Gs_u(y)= u^3y=y$ for a $3$rd primitive root $u$ of unity in $\F_q$ generates a cyclic group $\mG$ of order $3$. Let  $F=E^{\mG}$, where $E=\F_q(\fE)$ is the function field of $\fE$. The zero point $O$ is a totally ramified point with respect to the extension $E/F$.
In addition, the point $(0,\theta)$ and $(0,-\theta)$ are also  totally ramified since $\Gs_{u}(0,\theta)=(0,\theta)$ and $\Gs_{u}(0,-\theta)=(0,-\theta)$. Thus, the rest of $q+2\sqrt{q}-2$ rational points split completely. By Proposition \ref{prop:3.5}, there exists an optimal $q$-ary $[n=3\ell, k=2t+1,d=n-3t]$ locally repairable code with locality $2$ for any $0\le t<\ell\le \frac{q+2\sqrt{q}-2}3$. Note that in this case we have  $\frac{q+2\sqrt{q}-2}3=\left\lfloor\frac{q+2\sqrt{q}}3\right\rfloor$. This completes the proof.
\end{proof}

\begin{ex}{\rm
Let $q=4$ and $\F_{4}=\F_2(\Ga)$ with $\Ga^2+\Ga+1=0$.
We give an explicit construction of an optimal $4$-ary [6,3,3]-\LRC with locality $2$.
Let $E$ be the rational function field $\F_{4}(x,y)$ with $y^2+y=x^3$ and let $F=\F_4(y)$. Then $(y)_{\infty}=3O$ and $(x)_{\infty}=2O$.
It is easy to verify that $P_1=(1,\Ga)$, $P_2=(\Ga,\Ga)$, $P_3=(\Ga+1,\Ga)$ are lying over the same rational place $P_{y-\Ga}$ of $F$, and $P_4=(1,\Ga+1)$, $P_5=(\Ga,\Ga+1)$, $P_6=(\Ga+1,\Ga+1)$ are lying over the same rational place $P_{y-\Ga-1}$ of $F$.
Put \[V:=\left\{a_0+a_1y+b_0x:\; a_{0}, a_1, b_0\in\F_{4}\right\}.\]
Then the algebraic geometry code
$$
\begin{array}{lll}
C(\mP,V)=\left\{ \left(f(P_1), f(P_2), f(P_3),f(P_4),f(P_5), f(P_6)\right) :\; f\in V\right\}.
\end{array}
$$
is an optimal $4$-ary  $[6,3,3]$-\LRC with locality $2$.  Furthermore, a generator matrix of this code is computed as follows:
$$\left(\begin{array}{cccccc}
1 & 1 & 1 & 1 & 1 & 1  \\
\Ga & \Ga & \Ga & \Ga+ 1 &  \Ga+ 1  &  \Ga+ 1   \\
1 & \Ga & \Ga+1 & 1 & \Ga & \Ga+1  \\
\end{array}\right).$$
}\end{ex}

\begin{ex}{\rm Let $q=64$. By Theorem \ref{thm:1.1},  for any integers $t$ and $\ell$ with $0\le t< \ell \le 26$, there exists an optimal $64$-ary $[3\ell,2t+1,3\ell-3t]$-locally repairable code with locality $2$. In particular, for each integer $t$ with $1\le t \le 25$, there exists an optimal $64$-ary $[78,2t+1,78-3t]$-locally repairable code with locality $2$.
}\end{ex}

\begin{ex}{\rm Let $q=81$. By Theorem \ref{thm:1.1},  for any integers $t$ and $\ell$ with $0\le t< \ell \le 33$, there exists an optimal $81$-ary $[3\ell,2t+1,3\ell-3t]$-locally repairable code with locality $2$. In particular, for each integer $t$ with $1\le t \le 32$, there exists an optimal $64$-ary $[99,2t+1,99-3t]$-locally repairable code with locality $2$.
}\end{ex}

\begin{ex}{\rm Let $q=25$. By Theorem \ref{thm:1.1},  for any integers $t$ and $\ell$ with $0\le t< \ell \le 11$, there exists an optimal $25$-ary $[3\ell,2t+1,3\ell-3t]$-locally repairable code with locality $2$. In particular, for each integer $t$ with $1\le t \le 10$, there exists an optimal $25$-ary $[33,2t+1,33-3t]$-locally repairable code with locality $2$.
}\end{ex}

Now, we make use of Proposition \ref{prop:3.4} to get locally repairable codes of an odd locality $r$.
\subsection{Locality $r=3,5,7,11$ or $23$} The main purpose of this subsection is to prove Theorem \ref{thm:1.2}.

{\bf Proof of Theorem \ref{thm:1.2}:}
\begin{proof} We prove part (i) only. The similar arguments can be applied to other cases. By Proposition \ref{prop:3.4} and Remark \ref{rmk:2}, it is sufficient to show that there exists a maximal elliptic curve $\fE/\F_q$ defined by the Weierstrass equation \mbox{\eqref{eq:3}} and  a subgroup $\mG$ of $\Aut(\fE/\F_q)$ such that $|\mG|=4$ and the cardinality of the set $\{\Gs(x):\; \Gs\in\mG\}$ is $2$.

If $p=2$, by  Lemma \ref{lem:2.13}, there exists a maximal elliptic curve $\fE/\F_q$  defined by $y^2+y=x^3+\Ga$ for some $\Ga\in\F_q$ with $|\Aut(\fE/\F_q)|=24$. Then $\mG:=\{\Gs_{1,s,t}:\; s\in\F_2,\; \; t\in\F_4,\; t^2+t=s\}$ is a subgroup of $\Aut(\fE/\F_q)$ of order $4$, where $\Gs_{1,s,t}(x)=x+s$ and $\Gs_{1,s,t}(y)=y+s x+t$. Thus,  $\{\Gs(x):\; \Gs\in \mG\}=\{x,x+1\}$ has size $2$.

If  $p\equiv  3\pmod{4}$, by  Lemmas \ref{lem:2.14}, there exists a maximal elliptic curve $\fE/\F_q$  defined by $y^2+y=x^3+\Ga x$ for some $\Ga\in\F_q$ with $|\Aut(\fE/\F_q)|$ divisible by $4$. Then $\mG:=\langle\Gs_{u}\rangle$  is a subgroup of $\Aut(\fE/\F_q)$ of order $4$, where $\Gs_{u}(x)=u^2x$ and $\Gs_{u}(y)=u^4y$ and $u$ is a $4$th primitive root of unity in $\F_q$. Thus,  $\{\Gs(x):\; \Gs\in \mG\}=\{x,u^2x\}$ has size $2$.
\end{proof}

\begin{ex}{\rm Let $q=64$. Consider the elliptic curve over $\F_{64}$ defined by the equation $y^2+y=x^3$. It is maximal with $81$ rational points over $\F_{64}$. Consider the sugroup $\mG$ of $\Aut(\fE/\F_{64})$ given by $\mG:=\{\Gs_{1,s,t}:\; s\in\F_2,\; \; t\in\F_4,\; t^2+t=s\}$. Then $O$ is the unique ramified point and all other $80$ rational points split completely. These $80$ points can be taken as evaluation points and we get
an optimal $64$-ary $[3\ell,3t+2,3\ell-4(t-1)]$-locally repairable code with locality $3$ for all $1\le t<\ell\le 20$. In particular, for each integer $t$ with $1\le t \le 20$, there exists an optimal $64$-ary $[80,3t+2,80-4(t-1)]$-locally repairable code with locality $3$. For this example, $\ell$ can achieve $20$, while in Theorem \ref{thm:1.2}, $\ell$ is upper bounded by $\left\lfloor\frac{q+2\sqrt{q}-r-2}{r}\right\rfloor=18$.
}\end{ex}

\end{document}